\RequirePackage[l2tabu,orthodox]{nag}
\documentclass[a4paper,english]{scrartcl}
\pdfoutput=1 

\usepackage{todonotes,hyperref}

\usepackage[utf8]{inputenc}
\usepackage[T1]{fontenc}
\usepackage{lmodern} 

\usepackage{microtype}
\usepackage[within=none,figurename=Fig.,labelsep=period]{caption}

\usepackage{amssymb}
\usepackage{amsmath}
\usepackage{bm}   
\usepackage{mathtools}

\allowdisplaybreaks[1] 

\usepackage{tikz}
\usetikzlibrary{shapes}

\usepackage[inline]{enumitem}
\setlist[itemize]{label=$\circ$}
\setlist[description]{labelindent=\parindent}

\usepackage[amsmath,hyperref,thmmarks]{ntheorem}

\theoremseparator{.}
\newtheorem{theorem}{Theorem}
\newtheorem{lemma}[theorem]{Lemma}

\theoremstyle{plain}
\theorembodyfont{\upshape}

\newtheorem{definition}[theorem]{Definition}

\theoremstyle{nonumberplain}
\theoremheaderfont{\itshape}
\theorembodyfont{\upshape}
\theoremsymbol{\ensuremath{_\blacksquare}}
\newtheorem{proof}{Proof}

\DeclarePairedDelimiter\set{\{}{\}}
\DeclarePairedDelimiterX\setc[2]{\{}{\}}{\,#1 \;\colon\; #2\,}
\DeclarePairedDelimiterX\parenc[2]{\lparen}{\rparen}{\,#1 \;\delimsize\vert\; #2\,}

\newcommand{\cc}[1]{\ensuremath{\mathrm{#1}}}

\newcommand{\ccNP}{\cc{NP}}

\newcommand{\ccSW}{\cc{\#W[1]}}
\newcommand{\ccSP}{\cc{\#P}}

\newcommand{\N}{\mathbb{N}}
\newcommand{\Z}{\mathbf{Z}}

\newcommand{\GF}{\ensuremath{\mathbb{F}}}

\newcommand{\mat}{\mathrm{Mat}}

\newcommand{\rankbase}{\#\mathrm{Rank}\text{-}\mathrm{Bases}}
\newcommand{\nulbase}{\#\mathrm{Nullity}\text{-}\mathrm{Bases}}

\newcommand{\rank}{\mathrm{rk}}

\newcommand{\tree}{\#\mathrm{kTrees}}
\newcommand{\wtree}{\#\mathrm{WkTrees}}
\newcommand{\awtree}{\#\mathrm{ApexWkTrees}}
\newcommand{\wt}{\mathrm{WT}}

\usepackage{xparse}
\DeclareDocumentCommand{\restrict}{O{}}{\mathord{\restriction}_{#1}}

\title{Parameterized counting of trees, forests and matroid bases}

\author{Cornelius Brand, Marc Roth \\
Saarbr\"ucken Graduate School of Computer Science \\
Cluster of Excellence (MMCI), Saarland University}
\date{}

\begin{document}
 
\maketitle

\begin{abstract}
We investigate the complexity of counting trees, forests and bases of matroids from a parameterized point of view. It turns out that the problems of computing the number of trees and forests with $k$ edges are $\ccSW$-hard when parameterized by $k$. Together with the recent algorithm for deterministic matrix truncation by Lokshtanov et al. (ICALP 2015), the hardness result for $k$-forests implies $\ccSW$-hardness of the problem of counting bases of a matroid when parameterized by rank or nullity, even if the matroid is restricted to be representable over a field of characteristic $2$. We complement this result by pointing out that the problem becomes fixed parameter tractable for matroids represented over a fixed finite field.
\end{abstract}

\section{Introduction}
Since Valiant's seminal paper about the complexity of computing the permanent \cite{Valiant79}, counting complexity has advanced to a well studied subfield of computational complexity theory. By proving that the problem of counting the number of perfect matchings in a bipartite graph is complete for the class $\ccSP$ (the counting equivalent of $\ccNP$), he gave evidence that there are problems whose counting versions are inherently more difficult than their decision versions. For many interesting counting problems, it was shown to be $\ccSP$-hard to compute exact solutions. Therefore, several relaxations such as restrictions of input classes (see e.g. \cite{xia2007computational}) or approximate counting (see e.g. \cite{jerrum1989approximating}, \cite{dyer2010approximation}) were introduced. Another relaxation, the one this work deals with, is the analysis of the parameterized complexity of counting problems, which was introduced by Flum and Grohe in 2004 \cite{flumgrohe_counting}. They proved that, similar to classical counting complexity, there are problems whose decision versions are easy in the parameterized sense, i.e., which are fixed parameter tractable, but whose counting versions are most likely much harder. During the last years, much work has been done in the field of parameterized counting complexity. Important results are the proof of $\ccSW$-hardness for counting the number of $k$-matchings in a simple graph \cite{Curticapean2013}, and the dichotomies for counting graph homomorphisms \cite{grohe2007complexity,dalmau2004complexity} and embeddings \cite{curticapean_marx}. 

In this work, we analyze the parameterized complexity of counting trees and forests, parameterized by the number of edges, and counting bases of matroids, parameterized by rank or nullity. All of these problems are well studied in the classical sense. It is known that computing the number of all (labeled) trees is $\ccSP$-hard, even when restricted to planar graphs \cite{jerrum_trees}. The problem of counting all forests is $\ccSP$-hard even when restricted to $3$-regular bipartite planar graphs \cite{vertigan_welsh,gebauer_okamoto}. We will complement these results by showing that the problems of counting $k$-trees and $k$-forests are $\ccSW$ hard when parameterized by $k$. In both proofs, we reduce from the problem of counting $k$-matchings. Note that we cannot hope to achieve hardness results on planar graphs, as Eppstein showed that the problem of counting subgraphs of size $k$ is fixed parameter tractable on planar graphs \cite{eppstein2002subgraph}. To count the number of $k$-trees ($k$-forests) in planar graphs, we can just enumerate all possible trees (forests) consisting of $k$ edges and then apply Eppsteins algorithm for each of them, which ultimately yields an fpt algorithm.

Counting forests is very closely related to the problem of counting independent sets and bases in matroids: It is a well known fact that the $k$-forests of a graph $G$ correspond one-to-one to the $k$-independent sets of the binary matroid represented by the incidence matrix of $G$ over $\GF_2$ (see e.g. \cite{snook2012counting}). This shows already that the problem of counting $k$-independent sets in a binary matroid is $\ccSP$-hard in the classical sense. 
Moreover, it is known due to Vertigan that the problem of counting the bases of a binary matroid is $\ccSP$-hard (see e.g. \cite{vertigan1998bicycle}). On the other hand, a generalization of Kirchhoff's Theorem states that the number of bases of a regular matroid, i.e. a matroid that is representable over all fields, can be computed in polynomial time \cite{maurer_matroids}. 

We complement these results by proving that, when parameterized by rank or nullity, the problem of computing the number of bases of matroids is $\ccSW$-hard, even when restricted to matroids representable over fields of characteristic $2$. Most of the work of this proof is already contained in the hardness proof for $k$-forests. Having $\ccSW$-hardness of the problem of counting $k$-forests, we can use the recent algorithm for deterministic matrix truncation of Lokshtanov et al. \cite{Lokshtanov2015} to get the desired hardness result. We conclude with the fact that a hardness result for matroids represented over a fixed finite field, which includes the binary matroids, is very unlikely, since in this case, the problem of counting bases is fixed parameter tractable.

\section{Preliminaries}

\subsection{Parameterized counting complexity}
We begin with basic definitions of parameterized counting complexity,
following closely chapter 14 of the textbook \cite{flumgrohe}, which we recommend to the interested reader for a more comprehensive overview of the topic.
Our fundamental object of study is the following.
A \emph{parameterized counting problem} $(F,k)$ consists of a function $F:~\set{0,1}^\ast \rightarrow \N$
and a polynomial-time computable function $k:~\set{0,1}^\ast \rightarrow \N$, called the \emph{parameterization}.
A parameterized counting problem $(F,k)$ is called \emph{fixed-parameter tractable} if there is an algorithm for computing $F(x)$ running in time $f(k(x))\cdot |x|^c$ for some computable function $f:~\N \rightarrow \N$ and some constant $c > 0$ for all $x \in \set{0,1}^\ast$.
Let $(F,k)$ and $(F',k')$ be two parameterized counting problems.
Then, a function $R:~\set{0,1}^\ast \rightarrow \set{0,1}^\ast$ is called an \emph{fpt parsimonious reduction from $(F,k)$ to $(F',k')$} if
\begin{enumerate}
\item For all $x \in \set{0,1}^\ast$, $F(x) = F'(R(x))$.
\item $R$ can be computed in time $f(k(x))\cdot |x|^c$ or some computable function $f:~\N \rightarrow \N$ and some constant $c > 0$ for all $x \in \set{0,1}^\ast$.
\item There is some computable $g:~\N \rightarrow \N$ such that $k'(R(x)) \leq g(k(x))$ for all $x \in \set{0,1}^\ast$.
\end{enumerate}
An algorithm $A$ with oracle access to $F'$ is called an \emph{fpt Turing reduction from $(F,k)$ to $(F',k')$} if
\begin{enumerate}
\item $A$ computes $F$.
\item $A$ runs in time $f(k(x))\cdot |x|^c$ or some computable function $f:~\N \rightarrow \N$ and some constant $c > 0$ for all $x \in \set{0,1}^\ast$.
\item There is some computable $g:~\N \rightarrow \N$ such that for all $x \in \set{0,1}^\ast$ and for all instances $y$ for which the oracle is queried during the execution of $A(x)$, $k'(y) \leq g(k(x))$.
\end{enumerate}
The parameterized counting problem $\#k\textnormal{-}\mathrm{Clique}$ is defined as follows,
and is parameterized by $k$:
Given a graph $G$ and an integer $k$, compute the number of $k$-cliques in $G$.
The class $\ccSW$ is defined as the set of all parameterized counting problems $(F,k)$ such that there is an fpt parsimonious reduction from $(F,k)$ to $\#k\textnormal{-}\mathrm{Clique}$.
A parameterized counting problem $(F,k)$ is called $\ccSW$-\emph{hard} if there is an fpt Turing reduction from  
$\#k\textnormal{-}\mathrm{Clique}$ to $(F,k)$.

\subsection{Matroids}
A \emph{matroid} is a pair $M = (E,\mathcal{I})$ consisting of a finite ground set $E$ and $\mathcal{I} \neq \emptyset$, a family of subsets of $E$, that satisfies the following axioms:
\begin{enumerate}
\item $\mathcal{I}$ is downward closed, i.e. if $I \in \mathcal{I}$ and $I' \subset I$, then $I' \in \mathcal{I}$.
\item $\mathcal{I}$ has the exchange property, i.e. if $I_1,I_2 \in \mathcal{I}$ and $|I_1| < |I_2|$, then
there is some $e \in I_2 - I_1$ such that $I_1 \cup \set{e} \in \mathcal{I}$.
\end{enumerate}
Note that this entails $\emptyset \in \mathcal{I}$.
The elements of $\mathcal{I}$ are called \emph{independent sets}, and an inclusion-wise maximal element of $\mathcal{I}$ is called a \emph{basis} of $M$. The exchange property warrants that all bases have the same cardinality, and we call this cardinality the \emph{rank} of $M$, written as $\rank{M}$. Furthermore we define $(|E|-\rank{M})$ as the \emph{nullity} of $M$.

For a field $F$, a \emph{representation of $M$ over} $F$ is a mapping $\rho : ~ E \rightarrow V$, where $V$ is a vector space over $F$, such that for all $A \subseteq E$, $A$ is independent if and only if $\rho(A)$ is linearly independent. $M$ is called \emph{representable} if it is representable over some field.
If there is such a representation, we call $M$ \emph{representable over $F$}, or \emph{$F$-linear},
and it holds that $\rank(\rho) = \rank(M)$. In the following we will write $\GF_{p^n}$ for the field with $p^n$ elements.

Given a matroid $M=(E,\mathcal{I})$ we define the \emph{dual matroid} $M^*$ of $M$ as follows:
\begin{enumerate}
\item $M^*$ has the same ground set as $M$.
\item $B\subseteq E$ is a basis of $M^*$ if and only if $E \setminus B$ is a basis of $M$. 
\end{enumerate}
Given a representation of a matroid $M$ it is a well known fact that a representation of $M^*$ in the same field can be found in polynomial time. (see e.g. \cite{marx2009parameterized}).

In the following, \emph{all} matroids will be assumed to be representable,
and encoded using a representing matrix $\rho$ and a suitable encoding for the ground field.
Furthermore, we can, without loss of generality, always assume that $\rho$ has $\rank(M)$ many rows,
because row operations (multiplying a row by a non-zero scalar, and adding such multiples to other rows) do not affect linear independence of the columns of $\rho$. Hence, any $\rho'$ obtained from $\rho$ through row operations is a representation of $M$. In particular, by Gaussian elimination, we may assume all but the first $\rank(M)$ rows of $\rho$ to be zero.

\subsection{Graphs and matrices}
We consider simple graphs without self-loops unless stated otherwise. Given a graph $G$ we will write $n$ for the number of vertices of $G$ and $m$ for the number of edges. 
A $k$-\emph{forest} is an acyclic graph consisting of $k$ edges and a $k$-\emph{tree} is a connected $k$-forest. We say that two graphs $G_1=(V_1,E_1)$ and $G_2=(V_2,E_2)$ are \emph{ismorphic} if there is a bijection $\varphi: V_1 \rightarrow V_2$ such that for all $u,v \in V_1$
\[ \{u,v\} \in E_1 \leftrightarrow \{\varphi(u),\varphi(v)\} \in E_2 \]
A $k$-\emph{matching} of a graph $G=(V,E)$ is a subset of $k$ edges such that no pair of edges has a common vertex. The problem of counting $k$-matchings in a simple graph is known to be $\ccSW$-hard \cite{Curticapean2013}.

Given a graph $G$ with vertices $v_1,\dots,v_n$ and edges $e_1,\dots,e_m$ we define the (unoriented) \emph{incidence matrix} $M[G]\in \mat(n\times m, \GF_2)$ of $G$ as follows:
\begin{enumerate}
\item $M[G](i,j) = 1$ if $v_i \in e_j$ and
\item $M[G](i,j) = 0$ otherwise.
\end{enumerate}
It is a well known fact that a subset of columns of $M[G]$ is linearly independent (over $\GF_2$) if and only if the corresponding edges form a $k$-forest in $G$. \\
Let $F$ be a field and $M \in \mat(n\times m, F)$ be a matrix over $F$ with columns $c_1,\dots,c_m$. A $k$-\emph{truncation} of $M$ is a matrix $M^k$ of size $(k \times m)$ (not necessarily over the same field) with columns $c_1^k,\dots, c_m^k$ such that for every set $I \subseteq \{1,\dots,m\}$ of size at most $k$ it holds that 
$ \{c_i\}_{i \in I}$ is linear independent if and only if $\{c_i^k\}_{i \in I}$ is. Recently Lokshtanov et al proved that a $k$-truncation of a matrix can be done in deterministic polynomial time (see \cite{Lokshtanov2015}, Theorem 3.23).

\section{The parameterized complexity of counting trees and forests}
In this section we will prove that counting $k$-trees and $k$-forests is $\ccSW$-hard. The latter will be used to show hardness of counting matroid bases in fields of fixed characteristic.

\subsection{Counting $k$-trees}

\begin{definition}
	Given a graph $G$ and a natural number $k$, we denote the problem of counting all subgraphs of $G$ that are isomorphic to a $k$-tree as $\tree$.
\end{definition}

\begin{theorem}
	\label{thm:ktrees}
	$\tree$ is $\ccSW$-hard when parameterized by $k$.
\end{theorem}

For the proof, we show hardness of an intermediate problem.

\begin{definition}[Weighted $k$-trees]
	Given a graph $G=(V,E)$ with edge weights $\{w_e\}_{e\in E}$ and $k \in \N$, $\wtree$ is defined as the problem of computing
	\[ \wt_k(G) := \sum_{t \in T_k(G)} \prod_{e \in t} w_e \]
	where $T_k(G)$ is the set of all trees in $G$ consisting of $k$ edges.
\end{definition}

\begin{definition}
	\label{def:awtrees}
	Let $k$ be a natural number, $G = (V,E)$ be a graph with an apex $a \in V$, that is, a vertex that is adjacent to every other vertex, and edge weights $\{w_e\}$ such that $w_e = 1$ for all edges $e$ that are not adjacent to $a$ and $w_e = z$ for a $z \leq k$ otherwise. Then we denote the problem of computing $\wt_k(G)$ as $\awtree$. 
\end{definition}

\begin{lemma}
	$\awtree$ is $\ccSW$-hard.
\end{lemma}
\begin{proof}
	We reduce from the problem of counting $k$-matchings. Let $G=(V,E)$ be a graph with $n= |V|$. For $x$ and $z$ we construct the graph $G_{x,z}$ as follows:
	\begin{itemize}
		\item Add $x$ isolated vertices to $G$.
		\item Add a vertex $a$ to $G$ and connect $a$ to all other vertices, including the isolated ones (i.e., $a$ is an apex).
		\item Assign weights to the edges as follows: If $a \notin e$ then set $w_e = 1$. Otherwise set $w_e = z$.
	\end{itemize}
	The unweighted version of $G_{x,z}$ is just denoted by $G_x$. Furthermore we denote set of edges adjacent to $a$ as $E_x^a$.
	Now fix $x$ and compute for all $i\in \{0,\ldots,2k \}$ the value $P_i(x) = \wt_{2k}(G_{x,i})$ with an oracle for $\awtree$. This corresponds to evaluating the following polynomial in points $0,\ldots, 2k$:
	\begin{equation*}
	\begin{split}
	Q_x(z) &= \sum_{t \in T_{2k}(G_{x})} \prod_{e \in t} w_e\\
	& = \sum_{t\in T_{2k}(G_{x})} \prod_{\substack{e \in t\\ a\notin e}} 1 \cdot \prod_{\substack{e \in t\\ a\in e}} z\\
	& = \sum_{t\in T_{2k}(G_{x})} z^{|E_x^a \cap t |}
	\end{split}
	\end{equation*}
	Therefore we can interpolate all coefficients. In particular, we are interested in the coefficient of $z^k$ which is the number of trees of size $2k$ in $G_x$ such that exactly $k$ edges of the tree are adjacent to the apex $a$. We call these trees \emph{fair} and denote their number as $F_x$.
	Now consider an edge $e = \{v,a\}$ of such a fair tree $t$ that is adjacent to $a$. We call $e$ \emph{free} if $v$ is not adjacent to any other edge of $t$. Otherwise we call $e$ \emph{catchy} and we call the subtree $t^c$ of $t$ without free edges \emph{catchy} as well. (Note that $t^c$ is indeed a tree, not only a forest). Furthermore $t^c$ is a tree in $G_0$ as edges that are connected to isolated vertices cannot be catchy.
	
	Next let $\alpha_g$ be the number of catchy trees in $G_0$ of size $g+k$ such that $k$ edges of the tree are not adjacent to $a$. Then we claim
	\[ F_x = \sum_{g=1}^{k} \alpha_g \cdot \binom{n+x-(k+g)}{k-g} \]
	To prove this claim, consider a fair tree $t$ in $G_x$ with $g$ catchy edges. Then there are $\alpha_g$ possibilities for the catchy tree. As $t$ is fair, there are $k-g$ free edges to choose. Furthermore the catchy tree covers exactly $k+g$ vertices in $G_x$, that is, there are $n+x-(k+g)$ vertices that remain for the free edges. It follows that there are exactly $\alpha_g \cdot \binom{n+x-(k+g)}{k-g}$ possibilities. Note also that $\alpha_0$ is zero, as we are counting trees.\\
	Now we can evaluate $F_x$ for $x \in \{ 1,\ldots,k \}$ to solve for the $\alpha_g$. We have to show that the corresponding matrix has full rank:
	\[ A_{i,j} = \binom{i+n-k-j}{k-j} \]
	Then the $j$th column is an evaluation vector of the polynomial
	\[R_j(x) = \binom{x+n-k-j}{k-j} \]
	Furthermore the degrees of the polynomials are pairwise distinct and hence the polynomials and the evaluation vectors are also. It follows that $A$ has full rank, i.e., we can compute the coefficients.
	Finally consider $\alpha_k$: This is the number of catchy trees in $G_0$ with $k$ edges in $G_0$ and $k$ catchy edges. It follows that these trees correspond to $k$-matchings in $G$, more precisely, for every $k$-matching in $G$ there are $2^k$ such catchy trees. This concludes the proof.
\end{proof}

For the proof of Theorem~\ref{thm:ktrees} we need one further lemma.

\begin{lemma}
	$\awtree$ is fpt-turing reducible to $\tree$.
\end{lemma}
\begin{proof}
	Let $k$ be a natural number and $G=(V \cup \{a\},E \cup V \times \{a\} )$ be a graph with apex $a$ and edge weights as in Definition~\ref{def:awtrees}. The goal is to compute $\wt_k(G)$. Therefore we have to realize the edge weights that are not equal to $1$. Fortunately all of these edges are adjacent to the $a$. Let $z$ be the weight of these edges. If $z$ is zero, we just delete the apex and compute the number of $k$ trees in $G - \{a\}$. Otherwise we construct the graph $G^z=(V^z,E^z)$ from $G$ as follows:
	\begin{itemize}
		\item Delete the apex and the adjacent edges.
		\item Add apices $a_1,\ldots,a_z$ (including edges to all vertices in $G$).
		\item Add a vertex $a$ and edges $\{a,a_i\}$ for all $i \in [z]$.
	\end{itemize}

	Now we want to count the $k+z$ trees in $G^z$. As $z$ is promised to be bounded by $k$ we can do that in the fpt reduction. We will use the inclusion-exclusion principle to compute the number of $k+z$ trees in $G^z$ such that all edges $\{a,a_i\}$ for all $i \in [z]$ are met. We call these trees \emph{convenient}. Note that for any $v \in V$ a convenient tree can never contain two edges $\{v,a_i\}$ and $\{v,a_j\}$ for $i \neq j$ since this would induce the cycle $(v,a_i,a,a_j,v)$. 
	It follows that taking an edge $\{v,a\}$ with weight $z$ in a tree in $G$ is realised by taking one of the $z$ edges $\{v,a_i\}$ for $i \in [z]$ in a convenient tree in $G^z$.\\
	It remains to show how to use the inclusion-exclusion principle to enforce that the trees we count only consist of a chosen subset of edges. 
	\begin{lemma}
	\label{lem:incl_excl_tree}
	Let $G=(V,E)$ be a simple graph and $A \subseteq E$ a subset of edges of size $z$. Then the problem of counting the number of $(k+z)$ trees whose edges contain $A$ can be solved in time $O(2^z) \cdot \mathrm{poly}(|V|)$ if access to an oracle for $\tree$ is provided.
	\end{lemma}
\begin{proof}
Let $S$ be a subset of edges of $G$. We define $T_S$ as the set of all $(k+z)$ trees in $G$ that do not contain any edge in $S$. Note that we can compute $|T_S|$ by deleting all edges in $S$ and count the number of $(k+z)$ trees in $G$ by posing an oracle query. Furthermore it holds that \begin{align}
\label{eqn_incl_excl}
T_{S_1} \cap T_{S_2} = T_{S_1 \cup S_2}
\end{align} 
for any two subsets of edges $S_1$ and $S_2$. Let $T$ be the set of all $(k+z)$ trees in $G$, i.e., $T := T_\emptyset$. Now we can express the number of $(k+z)$ trees whose edges contain $A$ as 
\[|T \setminus \bigcup_{a \in A} T_{\{a\}}|\]
Using the inclusion-exclusion principle we get
\[ |T \setminus \bigcup_{a \in A} T_{\{a\}}| = |T| - \left( \sum_{\emptyset \neq J \subseteq A} (-1)^{|J|-1} ~|\bigcap_{a \in J}T_{\{a\}}|  \right)  = |T| - \left( \sum_{\emptyset \neq J \subseteq A} (-1)^{|J|-1} ~|T_J|  \right)\]
where the second equality follows from \eqref{eqn_incl_excl}. Finally we observe that there are exactly $2^{z}$ summands each of which can be computed by posing an oracle query for the corresponding $T_J$.
\end{proof}
Now Lemma~\ref{lem:incl_excl_tree} allows us to compute the convenient trees in $G^z$ in time $O(2^z) \cdot \mathrm{poly}(|V(G^z)|)$, given acces to an oracle for $\tree$. As $z$ is bounded by $k$, the construction of the Turing reduction is complete.
	
\end{proof}

\subsection{Counting k-forests}
\begin{definition}
	Let $G = (V,E)$ be a multigraph with edges labeled with formal variables $(w_e)_{e \in E}$.
	Then the \emph{multivariate forest polynomial} of $G$ is defined as
	\[
	F(G;(w_e)_{e \in E}) = \sum_{A \subseteq E \text{ acyclic }} \prod_{e \in A} w_e \,.
	\]
	The specialization of $w_e = x$ for all $e \in E$ and a fresh variable $x$ is called the
	\emph{univariate forest polynomial} of the graph and is simply denoted $F(G;x)$.
\end{definition}

The following was established by Curticapean \cite{Curticapean2013}:
\begin{theorem} \label{thm:kmatchings}
	Given a graph $G$ and a parameter $k$, it is $\ccSW$-hard to compute the number of 
	matchings of size $k$ in $G$.
\end{theorem}

Adding an apex, that is, a new vertex that is connected to all other vertices,
to a graph $G = (V,E)$ and labeling each of the new edges with a new variable $z$
makes the univariate forest polynomial into a bivariate one, namely the following:
\[
F(G';x,z) = \sum_{A \subseteq E \text{ forest }} x^{|A|} \cdot F(S_{A};z)
\]
where $G'$ is the described graph with an added apex and 
$S_{A}$ is the star-graph with $|c(A)|+1$ vertices and multiedges 
with multiplicity $|\kappa|$ to the vertex corresponding to the component $\kappa \in c(A)$.
In the following, $G$ will always be the original graph, and $G'$ will be the graph obtained in this way.

The corresponding forest polynomial $F(S_A;z)$ for a forest of size $|A| = n/2$ has the nice property that it vanishes at $z = -1$ whenever $A$ contains at least one component consisting only of a single vertex, since then the number of even and odd subsets of edges of $A$ is equal (evaluating a forest polynomial at $-1$ yields the difference between the number of even and odd forests in the graph).
On the other hand, if all components consist of exactly two points (i.e., every component is a single edge), then the polynomial is easily seen to evaluate to $(-1)^{n/2}$.
This shows that the coefficient of $x^{n/2}$ in $F(G';x,-1)$ is, up to sign, the number of perfect matchings in $G$. This proves that computing this coefficient is at least as hard as
counting perfect matchings of $G$.

It is now natural to ask whether the $k$-th coefficient in $F(G';x,z)$ can somehow be used to gain information on the number of $k$-matchings.
Clearly, the above property of the evaluation with $z = -1$ is not particularly useful in this regime: Every $k$-matching will leave at least one unmatched vertex in the graph, and therefore, the coefficient of $x^k$ is zero. More precisely, a forest with $k$ edges that results in exactly $n-2k$ single-vertex components is a $k$-matching, and every $k$-matching has this property.
We can leverage this property to generalize the case of $k = n/2$ to arbitrary $k$.

\begin{lemma} \label{lem:kcoeffhard}
	There is a polynomial-time Turing reduction from counting $k$-matchings in a graph $G$
	to computing the coefficient (which is an element of $\Z[z]$) of $x^k$
	of the bivariate forest polynomial $F(G';x,z)$ of the graph $G'$.
	In particular, this reduction retains the parameter $k$ and is thus a valid fpt-reduction.
\end{lemma}
\begin{proof}
	The $k$-th coefficient in $F(G';x,z)$ is given through
	\[
	C_k(z) = \sum_{A \in {E \choose k} \text{ forest}} \prod_{T \in c(A)} (1+|T|z).
	\]
	and in particular, if $M_k$ is the number of $k$-matchings in $G$, then
	after substituting $z \mapsto y-1$, we have
	\[
	C_k(y)/y^{n-2k} = M_k \cdot (2y-1)^{k} + R(y)
	\]
	for some polynomial $R(y)$ with the property that 
	\[
	R(y) = \sum_{i=1}^{k} q_i(y) y^{i}
	\]
	for some $q_i(y)$ that are either not divisible by $y$ or equal to zero.
	We see that $R(0) = 0$ so that $(C_k(y)/y^{n-2k})(0) = M_k \cdot (-1)^{k}$.
	All this can be easily done in polynomial time.
	
	Note how for $k=n/2$, this coincides with $C_k(y=0) = C_k(z=-1) = (-1)^{n/2} \cdot M_{n/2}$ because $y^{n-2k} = 1$.
\end{proof}

This proves that the coefficient of $x^k$ in the bivariate polynomial $F(G';x,z)$ is hard to compute.
We now want to show that this implies that the $k$-th coefficient of the univariate polynomial is
hard, and we want to do this by somehow computing the coefficient of $x^k$ in $F(G';x,z)$ from the $k$-th coefficient in suitable univariate forest polynomial.

We first show that although the degree of this coefficient polynomial (in the bivariate case) 
is not bounded by $f(k)$, but $\Omega(n)$, it suffices to know $O(k)$ coefficients (of the coefficient polynomial) in order to reconstruct the whole coefficient polynomial.

\begin{lemma} \label{lem:fewcoeff}
	We can reduce the computation of the coefficient of $x^k$ in $F(G';x,z)$ 
	to the computation of the first $k$ coefficients of fpt-many univariate forest polynomials
	on multigraphs.
\end{lemma}
\begin{proof}
	By definition,
	\[
	F(S_A;z) = \prod_{T \in c(A)} (1 + |T|z) \,,
	\]
	and since we are considering the coefficient of $x^k$, we know that $A$ will have
	between $n-k$ and $n-2k$ single-vertex components.
	In other words, $(1+z)^{n-2k}$ divides $F(S_A;z)$.
	Since $c(A)$ has at most $n-k$ components, $F(S_A;z)/(1+z)^{n-2k}$ is a polynomial of degree at most $k$ in $z$.
	It is therefore sufficient to know the coefficients of $1,\ldots,z^k$ in $F(S_A;z)$ to be able to reconstruct $F(S_A;z)$, namely as follows.
	Trivially, for some $f_i$, we have
	\[
	F(S_A;z) = \sum_{i=0}^{k} f_i z^i + \sum_{j=k+1}^{n-k} f_i z^i \,,
	\]
	and since $\{(1+z)^i\}_{i}$ is a basis, there are $f'_i$ with
	\[
	\sum_{i=0}^{k} f_i z^i = \sum_{i=0}^{n-k} f'_i (1+z)^i - \sum_{j=k+1}^{n-k} f_j z^j \,.
	\] 
	Since $(1+z)^{n-2k}$ divides $F(S_A;z)$, we know that $f'_i = 0$ for $i=0,\ldots,n-2k-1$,
	that is
	\[
	\sum_{i=0}^k f_i z^i = \sum_{i=n-2k}^{n-k} f'_i (1+z)^i - \sum_{j=k+1}^{n-k} f_j z^j \,.
	\]
	In particular, this equality must hold modulo $\langle z^{k+1} \rangle$, i.e.
	\[
	\sum_{i=0}^k f_i z^i = \sum_{i=n-2k}^{n-k} f'_i (1+z)^i + r\cdot z^{k+1}
	\]
	for some $r \in \mathbb{Z}[z]$.
	We can find $f'_i$ and $r$ satisfying this equality by solving a system of linear equations,
	and from this $r$ we can find $f_j$ for $j=k+1,\ldots,n-k$.
	It follows that we can reconstruct $F(S_A;z)$.
	
	This argument readily extends to a sum of forest polynomials: 
	If we know the sum of the first $k$ coefficients of $\{F(S_A;z)\}_{A \in \mathcal{A}}$ for some collection of forests $\mathcal{A}$ of the same size,
	then we can reconstruct the sum of all coefficients of this set of polynomials.
	This stems from the fact that all these polynomials have the required divisibility property,
	and hence their sum has it as well. 
	
	We now turn to the issue of reconstructing the coefficients of $x^i z^j$ for $i+j \leq 2k$ given access to the first $2k$ coefficients of the univariate forest polynomial.
	Letting $G'(a,b)$ be the graph obtained from $G'$ by $a$-thickening those edges labeled $x$, and $b$-thickening those edges labeled $z$, and labeling all these edges with $x$, we have
	\[
	F(G';ax,bx) = F(G'(a,b);x) \,.
	\]
	We write $c_{ij}$ for the coefficients of the monomial $x^iz^j$ in $F(G';x,z)$,
	and $d(a,b)_i$ for the coefficient of $x^i$ in $F(G'(a,b);x)$.
	For $t \leq 2k$, this equality implies on the monomial level that
	\[
	x^t \cdot \left(\sum_{i+j = t} a^i b^j c_{ij} \right) = 
	x^t \cdot d(a,b)_t \,.
	\]
	In particular, if $d(a,b)_t$ is known for $t \leq 2k$, 
	and we are looking to reconstruct $c_{ij}$ with $i+j\leq 2k$,
	then this identity provides us with $2k+1$ linear equations in the $2k \choose 2$ variables $c_{ij}$.
	Using $O(k)$ suitably chosen values for the pair $(a,b)$, we can solve the system for the $c_{ij}$
	and obtain the coefficients of $x^i z^j$ for $i+j \leq 2k$.
	In particular this implies that we can compute the coefficients of
	$x^k,x^kz,\ldots,x^kz^k$, which is nothing else than the first $k+1$ coefficients
	of the sum of the forest polynomials $F(S_A;z)$ for $|A| = k$.
	As argued above, this enables us to compute the remaining coefficients of this sum,
	i.e. the coefficient (an element of $\Z[z]$) of $x^k$ in $F(G';x,z)$.
\end{proof}

\begin{theorem}
\label{thm:hardness_k_forests}
	Counting $k$-forests in simple graphs is $\ccSW$-hard.
\end{theorem}
\begin{proof}
	Combining Theorem \ref{thm:kmatchings} with Lemma \ref{lem:kcoeffhard} and Lemma \ref{lem:fewcoeff} yields that computing the first $k$ coefficients of the univariate forest polynomial of multigraphs is $\ccSW$-hard. 
	Using Lemma 9 from \cite{ETHlpu}, this establishes hardness also for simple graphs.
\end{proof}

\section{Counting bases in matroids}

\begin{definition}
The problem of counting the number of bases of a matroid parameterized by its rank (nullity) is denoted as $\rankbase$ ($\nulbase$).
\end{definition}

\begin{theorem}
\label{thm:bases_hard}
The problems $\rankbase$ and $\nulbase$ are $\ccSW$-hard even when restricted to matroids representable over a field of characteristic $2$.
\end{theorem}

\begin{lemma}
\label{lem:bases_rank_hard}
The problem of counting $k$-forests in a simple graph is fpt Turing-reducible to the problem $\rankbase$ even when the matroid is restricted to be representable over a field of characteristic $2$.
\end{lemma}
\begin{proof}
Given a graph $G=(V,E)$ with $|V|=n$ and $|E|=m$ and a natural number $k$, we want to count the $k$-forests of $G$. Therefore we first construct the incidence matrix $M[G]\in \mat(n \times m, \GF_2)$ of $G$. Recall that the linearly independent $k$-subsets of columns of $M[G]$ correspond one-to-one to $k$-forests in $G$. In the next step we compute the reduced row echelon form of $M[G]$ by applying elementary row operations. As stated in the beginning, these operations do not change the linear dependency of the column vectors. Then we delete the zero rows which also does not change the linear dependency of the columns. We denote the resulting matrix as $M^{\mathrm{red}}[G]$. Now let $r$ be the rank of $M^{\mathrm{red}}[G]$ which equals the rank of $M[G]$. Note that $M^{\mathrm{red}}[G] \in \mat(r \times m, \GF_2)$. If $r<k$ then we output $0$ as $G$ does not have any $k$-forests in this case. Otherwise, we $k$-truncate $M^{\mathrm{red}}[G]$ in polynomial time by the deterministic algorithm of Lokshtanov et al. \cite{Lokshtanov2015} and end up in the matrix $M^k[G]\in \mat(k \times m,\GF_{2^{rk}})$. Note that the linear dependency of the column vectors is preserved, i.e., whenever columns $c_1,\dots,c_k$ are linearly independent in $M[G]$ they are also linearly independent in $M^k[G]$ and vice versa. Therefore the rank of $M^k[G]$ is at least $k$ since $M[G]$ has rank greater or equal $k$. As $M^k[G]$ has only $k$ rows it follows that the rank is \emph{exactly} $k$, i.e., $M^k[G]$ has full rank. Furthermore the number of linearly independent $k$-subsets of columns of $M^k[G]$ equals the number of $k$-forests in $G$. As the rank of $M^k[G]$ is full we conclude that the number of bases of the matroid that is represented by $M^k[G]$ equals the number of $k$-forests in $G$. Furthermore this matroid is by representable over $\GF_{2^{rk}}$ --- a field of characteristic $2$ --- by construction.
\end{proof}

\begin{lemma}
\label{lem:bases_nullity_hard}
The problem of counting $k$-forests in a simple graph is fpt Turing-reducible to the problem $\nulbase$ even when the matroid is restricted to be representable over a field of characteristic $2$.
\end{lemma}
\begin{proof}
We proceed as in the proof of Lemma~\ref{lem:bases_rank_hard}. Having $M^k[G]$ we construct its dual matroid $M^*[G]$ which can be done in polynomial time (see e.g. \cite{marx2009parameterized}). It holds that the number of bases of $M^*[G]$ equals the number of bases of $M^k[G]$. Furthermore the rank of $M^*[G]$ is $n-k$, i.e., its nullity is $k$, which concludes the proof.
\end{proof}

\begin{proof}[of Theorem~\ref{thm:bases_hard}]
Follows from Lemma~\ref{lem:bases_rank_hard}, Lemma~\ref{lem:bases_nullity_hard} and Theorem~\ref{thm:hardness_k_forests}.
\end{proof}

One might ask whether the same is true for matroids that are representable over finite fields. Due to Vertigan \cite{vertigan1998bicycle}, it is known that the classical problem of counting bases in binary matroids is $\ccSP$-hard. However, this does most likely not hold for the parameterized versions:

\begin{theorem}
\label{thm:bases_easy}
For every fixed finite field $\GF$, the problems $\rankbase$ and $\nulbase$ are fixed parameter tractable for matroids given in a linear representation over $\GF$.
\end{theorem} 

\begin{proof}
We give an fpt algorithm for $\rankbase$. An algorithm for $\nulbase$ follows by computing the dual matroid before.\\
Let $s$ be the size of the finite field, $M$ be the representation of the given matroid and let $k$ be its rank. We can assume that $M$ only has $k$ rows. (Otherwise we can compute the reduced row echelon form and delete zero rows which does not change the linear dependencies of the column vectors). If $M$ has only $k$ rows then there are at most $s^k$ different column vectors. Therefore we remember the muliplicity of every column vector and delete multiple occurences afterwards. We end up in a matrix with at most $s^k$ columns. Then we can check for every $k$-subset of columns whether they are linearly independent. If this is the case, we just multiply the multiplicities of the columns and in the end we output the sum of all those terms. The running time of this procedure is bounded by
\[ \binom{s^k}{k} \cdot \mathrm{poly}(n)\]
where $n$ is the number of columns of the matrix.
\end{proof}

\section*{Acknowledgements}
The authors wish to thank Holger Dell, Radu Curticapean and Markus Bläser for helpful comments on this work.

\bibliographystyle{plain}
\bibliography{matrixmod.bib}

\end{document}